\documentclass{amsart}
\usepackage{hyperref}
\usepackage{graphicx,amsfonts,amsmath,amsthm}

\newcommand{\T}{\mathbb{T}}

\newcommand{\D}{\mathbb{D}}
\newcommand{\I}{\mathbb{I}}

\newcommand{\B}{\mathbb{B}}
\newcommand{\E}{\mathbb{E}}
\newcommand{\X}{\mathbb{X}}
\newcommand{\Y}{\mathbb{Y}}
\newcommand{\N}{\mathbb{N}}
\newcommand{\Z}{\mathbb{Z}}
\newcommand{\Ss}{\mathbb{S}}

\newtheorem{theorem}{Theorem}
\newtheorem{assumption}{Assumption}

\newtheorem{lemma}[theorem]{Lemma}

\newtheorem{definition}{Definition}

\DeclareMathOperator*{\argmax}{arg\,max}
\usepackage{xcolor}

\newcommand{\PV}[1]{}
\def\dd{{\rm d}}

\theoremstyle{definition}

\theoremstyle{remark}

\numberwithin{equation}{section}

\begin{document}

\title{Some bidding games  converging  to their unique pure equilibrium}

\author{Benjamin Heymann}
\address{Criteo AI Lab, 32 rue Blanche 75009 Paris, FRANCE}
\curraddr{Criteo, 32 rue Blanche 75009 Paris, FRANCE}
\email{b.heymann@criteo.com}

\author{Alejandro Jofr\'e}
\address{Universidad de Chile, Beauchef 851, Edificio Norte – Piso 7
Santiago - CHILE}
\email{ajofre@dim.uchile.cl}

%    General info
\subjclass[2000]{Primary 91B26}

\date{}

\keywords{Auctions, Bayesian games, Bertrand's price Competition, Best Reply
Dynamics,  Bidding games,
Electricity markets, Lattices, Monotone comparative statics,  Pure Nash equilibria.}

\begin{abstract}
We introduce a class of  Bayesian bidding games for which we prove that the set of
pure Nash equilibria is a (non-empty) sublattice and 
we give a sufficient condition for
uniqueness that is often verified in the context of markets with
inelastic demand.
We propose a dynamic that converges to the extrema of the equilibrium set and derive
a scheme to compute  the extreme Nash
equilibria. 
 
\end{abstract}

\maketitle

\section{Introduction}
\label{sec:introduction}
The interaction between firms (or individuals) competing on price to maximize their profits in
an imperfect information environment constitutes a Bayesian game.
Very often most of the private information concerns the production costs of
the firms.  
Such  situations occur for instance in procurement auctions, commodity markets and
oligopolies.
It is then natural to ask about the Nash equilibria of such kind of
games. 
Do we have a guaranty of  existence? 
Is the equilibrium unique? 
Do we have an algorithm to compute it? 
Those are generically hard questions when the classical results of game theory are not
applicable. 
In this work we identify a class of games for which some of those questions can be answered using  tools from lattice theory.

An essential observation is that  a  bid increase by one firm is very often an incitation for  all other firms to increase their bids.
The exploitation of such monotonic behaviors, summarized in the notion
of strategic complementarity, is central in the study of many pricing games (and more generally in Bayesian games) and explains the intensive
use of lattice theory in a literature that  has many interesting results of pure Nash equilibria
existence. 
Those existence results, depending on the fixed point theorem over
which they are constructed,  differ by the underlying assumptions and the
additional information they provide on the equilibria set.

What follows builds on a rich literature. 
We now briefly review some of its major achievements.

\textbf{Optimization on lattices.}
The development of this work was strongly inspired by a book from
Topkis \cite{topkis1998supermodularity} in which the  author surveys some fundamental results of monotone comparative statics, in particular \cite{topkis1978minimizing,milgrom1994comparing,shannon1995weak,edlin1998strict}.
Monotone comparative statics   concerns setting where a parametrized
collection of optimal decisions are monotone in the parameter.

\textbf{Bayesian games and equilibrium existence.}
The class we introduce belongs to the large class of Bayesian games,
and more specifically has some strategic complementarity properties.

In \cite{vives1990nash}
Vives proves the existence of a pure Nash equilibrium for Bayesian games with general
action and type spaces  with   a Tarsky  fixed-point
theorem on lattice.  More precisely, the equilibrium set is a
non-empty complete lattice.
Payoffs need to be supermodular.

In \cite{athey2001single}, Athey shows the existence of a monotone
pure strategy equilibrium for finite  Bayesian games satisfying a
single-crossing condition.
Actions and types should be one dimensional.
The demonstration relies on
Kakutani's fixed-point theorem.
Athey generalizes the result to   a continuum action set when the
payoffs are continuous.

In \cite{mcadams2003isotone},
McAdams extends Athey results to a
multidimensional setting. 
Like in \cite{athey2001single}  the extension to  continuum action sets is obtained by taking the limit of
finite approximations   equilibria. 
McAdams, like Athey, uses a  single-crossing condition, combined  
with  a quasisupermodularity condition, and then applies the Glicksberg's  fixed point theorem.

In \cite{reny2004existence}, Reny and Zamir show that first
price, one unit auctions   with affiliated types and interdependent values have a pure equilibrium.
They combine a result by \cite{milgrom1982theory} with Athey
\cite{athey2001single} approach of limit of finite action grids.

Vives surveys in  \cite{vives2005complementarities} the use of
monotone comparative static tools for games with complementarities.
He points out some interesting properties of games with strategic
complementaries that are still valid in our framework: general
strategic spaces, existence of pure Nash equilibria, specific structure
of the set of equilibria, existence of a Best reply dynamics algorithm
to compute the extremal equilibria.
In the seventh section of the article, he discusses some aspects
specific to Bayesian games.

In \cite{van2007monotone},
Van Zandt and Vives
give a constructive proof of the existence of a pure Nash equilibrium for Bayesian games
satisfying a strategic complementarity condition.
The action and type
sets can be infinite dimensional.
The payoffs need to be supermodular and satisfy
some increasing difference property.
In addition,  they show that one can compute such equilibria by
best-reply iteration.
We propose a different  approach that turns out to work on the toy example we provide even for instances that violate the increasing difference property assumption.

In \cite{reny2011existence} Philip J. Reny generalizes the results of
Athey \cite{athey2001single} and McAdams \cite{mcadams2003isotone} on
the existence of monotone pure strategy equilibria in generic Bayesian
games. While Athey and McAdams proofs rely on the  convexity
of the  best reply sets, 
Reny uses a fixed point theorem that relies on the notion of
contractibility.
He shows that the result applies when the payoff function is weakly
quasi supermodular and satisfies a weak single crossing condition, and
concludes that his result is strictly more general than
\cite{athey2001single} and \cite{mcadams2003isotone}.
In particular, the result can  be applied when 
 type and action sets are infinite dimensional.
 The payoff functions need to be continuous in the actions.

\textbf{Equilibrium computation, best reply dynamics and fictitious
play.}
Many  approaches to compute an equilibrium consist in  mimicking the
behaviors of the players when the game is
repeated sequentially. As time goes on, each player takes a decision
based on the previous iterations of the game.
Basically, an approach is characterized by 
the  \emph{memory} of the players (do they remember a joint probability of actions,
marginals\ldots) and the way they  \emph{choose} the next
action.
The two historical approaches are the Cournot's t\^atonnements and Brown's fictitious play \cite{brown1949some}. In the standard Cournot's
t\^atonnements (or Best Reply dynamics), actions
are taken as best replies against the last actions of the other
players. 
In fictitious play, actions are taken as best replies against
 the average of the  other players past
actions.
There are no general results of fictitious play~\cite{10.2307/1969530} convergence for games
with complementarities. Yet one may consult
\cite{berger2009convergence,berger2008learning,milgrom1990rationalizability,milgrom1991adaptive}.
Vives discusses in  \cite{vives2005complementarities} 
the t\^atonnements  in a context very close to ours.
The convergence is derived from some monotonicity properties but
as he  points out, convergence cannot be ensured for an arbitrary
starting point.
Observe that most of these results are related to matrix games. 
We propose an alternative approach to fictitious play and Best Reply
dynamics  for  a class of Bayesian games to which the electricity
market introduced in \cite{NicolasFigueroaAlejandroJofrBenjaminHeymann} belongs. \\

\textbf{Auctions, public good games and Cournot equilibrium}
 Existence, uniqueness and computation of equilibrium are very important in games that model economic situation such as Cournot competition~\cite{gaudet1991uniqueness}, public good games~\cite{bayer2021best} and auctions~\cite{marshall1994numerical,lebrun1999first,fibich2003asymmetric,maskin2003uniqueness,Gayle2008,Fibich2011}.
 Our setting is quite close to first price auctions (in fact, the running example provided thereafter becomes a first price auction when the parameter $r$ goes to zero). However, it is notable that most results on first price auctions exploit the fact the first order condition on the optimal bid takes the form of a differential equation, which is something we do not have in the present paper.  
\\

\textbf{Concave games}
 Following the seminal work by Rosen~\cite{rosen1965existence}, there is a very important stream of literature that relies on concavity arguments~\cite{arrow60,rosen1965existence,Goodman1965NoteOE,gabay1980uniqueness,chenault1986uniqueness}.
As illustrated by the nature of the assumptions we require, our results are in essence different from this stream of literature.
For instance, our existence result does not rely on any smoothness assumptions.
The uniqueness result is for continuum of types --- such result could not be recovered by simply  discretizing the type set  and then taking the limit---.
The  fact that  the  dynamic lives in a lattice allows us to work directly with  continuum of types. 
Also, the example provided in this article does not satisfy the coercivity assumption required in~\cite{gabay1980uniqueness}, nor does it satisfy the second "convexity" assumption of~\cite{Goodman1965NoteOE}.
 \\
 
\textbf{Contributions.}
We identify a natural class of games that model sellers competing on price to increase their market share and maximize their profit. The main primitive of such game is the demand function associated to each seller. Such demand function can be, for instance, the result of an optimization problem solved by some hypothetical buyers. 
We use the specific structure of the set of equilibria (nonempty complete lattice) to derive a uniqueness sufficient condition. 
The condition can be interpreted as an oligopolistic price competition with inelastic demand. 
We also derive  a numerical scheme to compute the extrema of the equilibrium set. 
This scheme proved more robust than best reply iterations on the example that motivated this study.
Our approach underlies 
the geometry of the equilibrium set for this kind of games,  and the specific shape of gradient flow dynamic (monotonic iterate),  which allowed us  to design a scheme that maintain the monotony of the best replies,  even when the hypothesises were not fully satisfied. 
By contrast to other authors who used  discrete approaches to extend their result to Bayesian games~\cite{CEPARANO2017154} our perspective allows us to work directly on  games with a continuum of types.
\\
In the next section, we introduce the game and present our main
results. 
We illustrate those results on an example in  \S\ref{sec:examples}.
\S\ref{sec:existence}, \S\ref{sec:uniqueness} and \S\ref{sec:scheme} are dedicated to the
proofs of the three main results: the existence of a Nash equilibrium,
the uniqueness of the equilibrium  and
a  convergence of a t\^atonnement dynamics to the equilibrium.

\section{Game presentation and main results }
\subsection{Definitions}

\begin{definition}[Least Upper Bound ($\vee$)
and Greatest Lower Bound ($\wedge$) , see \cite{topkis1998supermodularity}]
Let $\X$ be  a partially ordered set, $\X'\subseteq\X$. 
We say that $\bar{x}\in \X$ is the least upper bound of $\X'$ when 
$\forall x\in\X, \quad \bar{x} \leq x \iff  x'\leq x\  \forall x'\in\X' $.
We say that $\underline{x}\in \X$ is the greatest  lower bound of
$\X'$ when 
$\forall x\in\X \quad x \leq \underline{x} \iff  x\leq x',\ \forall x' \in\X' $.
For $(x,y)\in\X^2$, we   denote by  $x\wedge y $ and $x\vee y$ 
the greatest lower bound and the least upper bound of the pair $\{x,y\}$. 
\end{definition}
\begin{definition}[Lattice, Sublattice, Complete Lattice, see \cite{topkis1998supermodularity}]
A partially ordered set $\X$ is a lattice iff  it contains a least upper
bound and a greatest lower bound for each pair of its elements.
A subset $\X'\subseteq \X$  is a sublattice if  it contains a least upper
bound and a greatest lower bound for each pair of its elements.
A lattice  in which each nonempty subset has a greatest lower bound
and a least upper bound is complete.
\end{definition}
We will use the  notion of increasing function in lattice, which
differs from the usual definition. Observe that one may also encounter
the term isotone in the literature.

\begin{definition}[Increasing]
We say that a function $f$ from two ordered sets  is \emph{increasing} if for all $x\leq y$,
$f(x)\leq f(y)$.
\end{definition}

\subsection{Game Presentation}

\label{sec:problem-presentation}
\subsubsection{Notations}

The game $(\I,\T,\B,\Sigma,K,p)$ consists in  a set of \emph{players}
$\I=1\ldots n$, $n\in \N$. 
For each player, there is
 a set of \emph{types} $\T^i$ and a set of  \emph{bids} (or
 \emph{action}) $\B^i$. 
Types and bides   
are included in 
 a compact interval $[c_*,c^*]$ (where $c_* > 0$) such that
   $\T\subset \B$. 
The  \emph{strategies}   are  applications
 from $\T^i$ to $\B^i$.
For each player, we denote by $\Sigma^i$ his strategy set.
For any $i\in \I$, the \emph{demand response}   $K^i$ is a   function
from the bid  set $\B$ to $[0,K^+]$, where $K^+>0$.
We generically denote by  $\sigma^i$ the elements of the
strategy set $\Sigma^i$,  $c^i$ the elements of $\T^i$ (because it can be interpreted as a
production \emph{cost} and we want to avoid any confusion with the
time variable), and $b^i$ the  bids, elements of $\B^i$. 
We use the standard notation of game theory  $-i$ to refer to all but player
$i$. 
Last but not least, $p^{-i}$ is a probability  density of support $\T^{-i}$.

For a strategy profile $\sigma^{-i}$, the
expected ex-ante payoff  $\Pi^i$  of
player $i$ of type $c^i$ bidding $b^i=\sigma^i[c^i]$, is 
\begin{equation*}
  \Pi_{\sigma^{-i}}^i(b^i,c^i) = \int_{c^{-i}\in T^{-i}} \pi^i(b^i,c^i,\sigma^{-i}[c^{-i}]) p^{-i}(c^{-i}) \dd c^{-i},
\end{equation*}
with the \emph{payoff} $\pi^i$ defined for
$(b^i,c^i,\sigma^{-i},c^{-i})\in \B^i\times \T^i\times \Sigma^{-i}
\times \T^{-i}$ by
\begin{equation*}
\pi^i(b^i,c^i,\sigma^{-i}[c^{-i}])=(b^i  - c^i ) K^i(b^i,\sigma^{-i}(c^{-i})).
\end{equation*}
In this expression, $K^i$ can be interpreted as the quantity $i$ is asked
to provide for a \emph{marginal price profile} $b^i$ when the other
players bid the marginal prices
$\sigma^{-i}[c^{-i}]$. So  the integrand is the profit of this player
if he has a marginal production cost $c^i$. 
The kernel $K^i$ corresponds to the market (or auctioneer) response to
the bids.  We assume the Kernel to be continuous.
In what follows we assume $\pi^i(b^i,c^i,\sigma^{-i}[c^{-i}])$ Lebesgue measurable with respect
to $c^{-i}$ for all $(b^i, c^i,\sigma^{-i})\in \B^i\times \T^i \times\Sigma^{-i}$.

\begin{definition}[Best Reply ]
 We denote by $BR^i$ the Best Reply set-valued mapping from
 $\Sigma^{-i}$ to the subsets of $\Sigma^i$ such that 
for any  $\sigma^{-i} \in \Sigma^{-i}$,
\begin{equation*}
BR^{i} (\sigma^{-i})=  \{\beta \in \Sigma^i: \quad \forall
(c,\sigma)\in\T^i \times\Sigma^i, \quad\Pi_{\sigma^{-i}}^i(\beta[c^i],c^i) \geq \Pi_{\sigma^{-i}}^i(\sigma^i[c^i],c^i) \}
  \end{equation*}
\end{definition}

\begin{definition}[Pure  Nash Equilibrium]
A strategy profile $\sigma\in \Sigma$ is a Pure Nash Equilibrium if for any $i\in\I$, 
 $\hat{\sigma}^i \in \Sigma^i $ and $c^i\in\T^i$
\begin{equation*}
  \Pi_{\sigma^{-i}}^i(\sigma^i[c^i],c^i) \geq \Pi_{\sigma^{-i}}^i(\hat{\sigma}^i[c^i],c^i) 
  \end{equation*}
\end{definition}

We   use the partial order $\leq_{\Sigma^i}$ on $\Sigma^i$ defined by
\begin{equation*}
\forall (\sigma_1,\sigma_2) \in \Sigma^i, \quad (\sigma_1\leq_{\Sigma^i}\sigma_2)\quad \mbox{iff} \quad
(\forall c\in\T^i ,\quad \sigma_1(c)\leq \sigma_2(c)).
\end{equation*}
We denote by $\leq_\Sigma$ the induced product order on $\Sigma$.

\begin{assumption}[Kernel Monotonicity]
\label{assumption:kernelmonotonicity}
For any $i\in\I$  $K^i(b^i,b^{-i})$, is increasing in $b^{-i}$ and
decreasing in $b^i$. 
\end{assumption}
The \emph{Kernel Monotonicity} assumption corresponds to the fact that the bidding occurs in a
competitive setting, and the demand tends to go to the cheapest
bidder.

\begin{assumption}[Strict Increasing Differences]
 \label{assumption:IDP}
For any $i\in\I$, $c\in\T^i$,  set $\pi^i_c= \pi^i(.,c,.)$. Then
$\pi^i_c$ satisfies the Strict Increasing Differences Property:
\begin{eqnarray*}
\forall (b_1,b_2,b^{-i}_1,b^{-i}_2)\in \B^{i}\times\B^i\times
\B^{-i}\times \B^{-i} \mbox{ such that }
b_1\leq b_2\mbox{ and } b_1^{-i} < b_2^{-i}:\\
\pi^i_c(b_2,b_1^{-i})-\pi^i_c(b_1,b_1^{-i})< \pi^i_c(b_2,b_2^{-i})-\pi^i_c(b_1,b_2^{-i}).
\end{eqnarray*}

 \end{assumption}

\subsection{Main results} 
\begin{theorem}[Existence of a Pure Nash Equilibrium]
\label{th:existence}
The set of
pure Nash equilibra is a nonempty complete lattice.
\end{theorem}

\begin{theorem}[Uniqueness Sufficient Condition]
\label{th:uniqueness}
If
\begin{itemize}
\item for any $\alpha>0$, $(x,y)\in\B$, 
\begin{equation}
\label{as:scaling}
K(\alpha x,\alpha y) = K(x,y),
\end{equation}
\item for any $\sigma\in \Sigma$  equilibrium
strategy profile, $i\in \I$, and  $(c_1,c_2)\in [c_*,c^*]\times ]0,c^*]$ such that
$c_1>c_2$  
\begin{equation}
\label{eq:strictincrease}
\beta_1 >\beta_2, \forall
(\beta_1\beta_2)\in \argmax_{b\in\B^i} \Pi_{c_1}^i(b,\sigma^{-i})\times \argmax_{b\in\B^i} \Pi_{c_2}^i(b,\sigma^{-i})
\end{equation}
\item for any $\sigma\in \Sigma$  equilibrium
strategy profile, $i\in \I$, $\inf BR^i(\sigma^{-i})(c)$
and $\sup BR^i(\sigma^{-i})(c)$ are continuous
\item  for any $(\sigma_1,\sigma_2)\in \Sigma$  equilibrium
strategy profile
for all $(i,c)\in \I\times\T$,
\begin{equation*}
\label{eq:alpha}
\sigma_2^i(c^*)\frac {\sigma_1^i(c)} {\sigma_2^i(c)}\leq b^*
\end{equation*}

\end{itemize}
then the set of pure Nash equilibria is a singleton.
\end{theorem}

The first item means  that if every player multiplies his
bid by the same constant, then the resulting allocation does not
change. This should be satisfied  with inelastic demand.

\begin{theorem}[Converging dynamic]
\label{th:scheme}

Assume:
\begin{itemize}
\item $b\in \B^i \rightarrow \Pi^i_{\sigma^{-i}}(b,c)$ is
  $C^2$ and 
  $ b \rightarrow\partial_{b}\Pi^i_{\sigma^{-i}}(b,c)$ is
 uniformly
Lipschitz  for all $(c,\sigma)\in \T^i\times\Sigma^{-i}$.  
\item $\sigma_*(\T)$ is  included  in  $]b_*,b^*[^{\I}$, where $\sigma_*$ is the smallest equilibrium's strategy profile. 
\item For any $(\sigma^{-i},c) \in \Sigma^i\times\T^i$, $\Pi^i_{\sigma^{-i}}(b,c)$ is concave in  $b$.
\end{itemize}
Then the solution to  the system of differential equations 
\begin{eqnarray*}
\label{system}
\forall (i,c)\in \I\times \T^i \quad  \partial_t \sigma^i(c,t)=\partial_{b}\Pi^i_{\sigma(,t)^{-i}}(\sigma^i(c,t),c)
 \\
 \sigma^i(c,0)=c
\end{eqnarray*}
converges to  the smallest equilibrium strategy profile
 $\sigma_*$ as $t$ goes to $+\infty$.
\end{theorem}

\section{Application}
\label{sec:examples}

Consider a simple geographical
electricity market with two nodes (Node 1 and Node 2). 
The nodes are connected by a line  through which  electricity can be
sent. There is a known (inelastic) demand $d$ at each node. 
We assume marginal prices to be constant, within a compact $[b_*,b^*]$.
We consider that there is one producer at each node, namely
$a_1$ and $a_2$. 
An independent operator has to allocate the production  to meet the
demand at each node and minimize the total cost paid to the
producers. 
When a quantity $h$ of  electricity is sent through the line, $rh^2$
is lost in the process (Joule effect). 
The players of the Bayesian game are the two producers, who want to maximize
their expected profit (we say \emph{expected} because they do not know the other player
production cost).
Solving the independent operator problem, we get
\begin{equation*}
K^i(b^i,b^{-i})= \begin{cases} F(b^i,b^{-i}) &\mbox{if } F(b^i,b^{-i})
  \geq 0
  \mbox{ and } F(b^{-i},b^{i}) \geq 0\\ 
0  & \mbox{if }  F(b^i,b^{-i})
  < 0\\
\bar{q}
&\mbox{if }F(b^{-i},b^{i}) < 0, \\
\end{cases}
\end{equation*}
where 
\begin{equation*}
F(x,y) = d+\frac{1}{2r}\left(\frac{x-y}{x+y}\right)^2 - \frac{1}{r}\frac{x-y}{x+y}\quad
\mbox{ and }\quad
\bar{q}= 2\frac{1-\sqrt{1-2dr}}{r}.
\end{equation*}
Therefore Assumption \ref{assumption:kernelmonotonicity} is 
satisfied.
Observe that the increasing difference property is not 
satisfied  everywhere.

In the following, we assume that $F(b^*,b_*)\geq 0$. Therefore the
corner solutions are not to be considered and  $K^i(b^i,b^{-i})=
F(b^i,b^{-i})$. Moreover, we assume $b^* <2 b_*$.
Then the payoff writes $\pi^i_c(b^i,b^{-i}) = (b^i-c)K^i(b^i,b^{-i})$.
Therefore 
\begin{equation*}
\partial_{xy}\pi^i_c(x,y) = \frac{4y}{r(x+y)^4}\left(     x(2y-x)+c(2x-y)       \right)>
0
\end{equation*}
Therefore the  strict increasing differences condition \ref{assumption:IDP}
is satisfied.
So Theorem \ref{th:existence} applies.

Next it  is clear  that the scaling property \eqref{as:scaling} of Theorem
\ref{th:uniqueness} is satisfied. 
Using 
\begin{equation*}
\partial_{xc}\pi^i_c(x,y) = \frac{4y^2}{r(x+y)^4}>
0
\end{equation*}
we show as in \cite{van2007monotone} that \eqref{eq:strictincrease} is satisfied.

Observe that 
\begin{equation*}
\partial_{xx}\pi^i_c(x,y) = \frac{4y}{r(x+y)^4}\left(     x - 3 c - 2
  y       \right) < 0
\end{equation*}
 Therefore $\pi$ is strictly  concave with respect to its first variable, and by
 integration, so is $\Pi$.  
Therefore by Berge's theorem, the best reply is continuous in $c$. 
To show that \eqref{eq:alpha}, we first observe that in full
information, symmetric  setting, if the cost is $c$ for both players, then the best
reply is  $\frac{c}{1-2rd}$. We combine this observation with the
monotonicity of the best reply with respect to the type and the
opponent strategy to conclude that any optimal bid $b$ should satisfy
\begin{equation*}
b\in [\frac{c_*}{1-2rd}, \frac{c^*}{1-2rd}]
\end{equation*} 
Therefore condition \eqref{eq:alpha} is satisfied for $b^*$ large
enough and Theorem
\ref{th:uniqueness} applies.

We have already checked that all condition to apply Theorem
\ref{th:scheme} were satisfied.

\section{Existence of a Nash Equilibrium}
\label{sec:existence}

\subsection{General Preliminary Results}

\begin{definition}[Strict Single crossing property, see \cite{topkis1998supermodularity}]
\label{def:sscp}
Let $\X$, $\Y$ and $\Z$ be partially ordered set, let $f(x,z)$ be a
function of a subset $\Ss $ of $\X\times \Z$ into $\Y$, then $f(x,z)$
satisfies the strict single crossing property in $(x,z)$ on $\Ss $ if for all $x_1$
and $x_2$ in $\X$ and $z_1$, $z_2$ in $\Z$ with $x_1<x_2$, $z_1<z_2$ and
$(x_1,x_2)\times (z_1,z_2)$ being a subset of $\Ss$, $f(x_1,z_1) \leq
f(x_2,z_1)$ implies $f(x_1,z_2) <
f(x_2,z_2)$.
\end{definition}

\begin{lemma}
\label{lemma:SSCP}
The application
$(b,\sigma^{-i}) \rightarrow\Pi^i_{\sigma^{-i}}(b,c)$ satisfies the Strict Single Crossing Property.
\end{lemma}
\begin{proof}
Let $b_1<b_2\in\B$ and $\sigma_1^{-i}<\Sigma_2^{-i}\in\Sigma^{-i}$ such that 
$\Pi^i_c(b_1,\sigma_1^{-i})\leq \Pi^i_c(b_2,\sigma_1^{-i})$.
By increasing differences, for any $c^{-i}\in \T^{-i}$, we have 
\begin{equation*}
\label{eq:intgrandIneq}
\pi^i_c(b_2,\sigma_1^{-i}(c^{-i}))-\pi^i_c(b_1,\sigma_1^{-i}(c^{-i}))<
\pi^i_c(b_2,\sigma_2^{-i}(c^{-i}))-\pi^i_c(b_1,\sigma_2^{-i}(c^{-i})),
\end{equation*} so
multiplying by $p^{-i}(c^{-i})$, and integrating,
we get
\begin{equation*}
0\leq\Pi^i_c(b_2,\sigma_1^{-i})- \Pi^i_c(b_1,\sigma_1^{-i})< \Pi^i_c(b_2,\sigma_2^{-i})- \Pi^i_c(b_1,\sigma_2^{-i}),
\end{equation*}
where the first inequality comes from the hypothesis.
\end{proof}
\begin{definition}[Quasisupermodularity,  see
  \cite{topkis1998supermodularity}]
\label{def:quasisupermodularity}
Let $\X$ be a lattice, $\Z$ a partially ordered set and $f$ a
function
from $\X$ to $\Z$, then we say that $f$ is quasisupermodular if for all
$x_1$ and $x_2$ from $X$, $f(x_1 \land x_2) \leq f(x_1)$ implies $f(x_2)\leq f(x_1
\lor x_2)$ and  $f(x_1 \land x_2) < f(x_1)$ implies $f(x_2)< f(x_1
\lor x_2)$.
\end{definition}
\begin{lemma}[Quasi Supermodularity]
\label{lemma:QSM}
 For any $i\in \I$, $c \in \T$, $\sigma^{-i}\in \Sigma^{-i}$, 
$b\rightarrow\Pi^i_{\sigma^{-i}}(c,b)$ is quasisupermodular.
\end{lemma}
\begin{proof}
Trivial since we are in a monodimensional setting.
\end{proof}
\subsection{Existence}
\label{section:existence}
We will need the following result: 
\begin{theorem}[Increasing Optimal Strategies
  (see \cite{topkis1998supermodularity}) page 83]
\label{th:increasingDecision}
Suppose that $\X$ is a lattice, $\Z$ is a partially ordered set, $\Ss_z$ is
a subset of $\X$ for each $ z$ in $\Z$, $\Ss_z$ is increasing in $z$ on $\Z$,
$f(x,z)$ is quasisupermodular in $x$ on $X$ for each $z$ in $\Z$, and
$f(x,z)$ satisfies the strict single crossing property in $(x,z)$ on $\X\times\Z$. If $z_1$ and $z_2$ are in $\Z$, $z_1<z_2$, $x_1$ is in
$\argmax_{x\in \Ss_{z_1}}f(x,z_1)$ and  $x_2$ is in
$\argmax_{x\in \Ss_{z_2}}f(x,z_2)$, then $x_1 \leq x_2$. Hence if one picks
any $x_z$ in $\argmax_{x\in \Ss_z} $ for each $z$ in $\Z$ with argmax non
empty, then $x_z$ is increasing in $z$ on $\{z:z\in \Z,\quad
\argmax_{x\in\Ss_z}f(x,z) \quad \mbox{non
  empty } \}$.
\end{theorem}
\begin{definition}[Induced Set ordering]
Let $\X$ be a lattice, $\X_1$ and $\X_2$ two non empty subsets of $\X$. 
We say that  $\X_1\sqsubseteq
 \X_2 $ iff 
for any $(x_1,x_2) \in \X_1\times\X_2$, 
$x_1\wedge x_2 \in \X_1$ and 
$x_1 \vee x_2 \in \X_2$.
\end{definition}

Combining Lemma \ref{lemma:SSCP}, Lemma \ref{lemma:QSM} with Theorem 
\ref{th:increasingDecision}, we get:
\begin{lemma}
\label{lemma:increasingBR}
For any $i\in \I$,  for any $(c,\sigma^{-i}_1,\sigma^{-i}_2)\in\T^i\times
  (\Sigma^{-i})^2$, such that  $\sigma^{-i}_1<\sigma^{-i}_2$,
for any $(\beta_1,\beta_2)\in BR^i(\sigma^{-i}_1)\times BR^i(\sigma^{-i}_2)$
\begin{equation*}
\label{eq:increaingBR}
\beta_1(c) \leq\beta_2(c)
\end{equation*}
In particular,  $BR^i$ is increasing in the induced set ordering on
$\Sigma^i$.
In addition,
 for any $\sigma^{-i} \in \Sigma^{-i}$,
$BR^i(\sigma^{-i})$ is a complete sublattice.

\end{lemma}
\begin{proof}
By continuity and compactness, for any $c\in\T^i$, $\argmax_b
\Pi^i_c(b,\sigma^{-i})$ is nonempty. 
Since $\Pi^i_c$ is quasisupermodular (Lemma \ref{lemma:QSM}) and
satisfies the strict single crossing property (Lemma \ref{lemma:SSCP}),
 by Theorem
\ref{th:increasingDecision},\eqref{eq:increaingBR} is satisfied for
any $c$. So any selection of $BR^i$ is increasing in $\sigma^{-i}$.
Therefore $BR^i$ is increasing in the induced set ordering.
Since $\Pi$ is continous in $b$ (by continuity of the integrand) and
supermodular in $b$ (since $b$ is monodimensionnal), by corollary
2.7.1 of \cite{topkis1998supermodularity}, $BR^i(\sigma^{-i})[c]$ is a
complete sublatice of $\B^i$, therefore $BR^i(\sigma^{-i})$ is a
complete sublatice  of $\Sigma^i$.
\end{proof}

\begin{theorem}[Tarsky fixed point~\cite{tarski1955lattice}, see
  \cite{topkis1998supermodularity} Theorem 2.5.1]
\label{theorem:fixedpoint}
Suppose that $\X$ is a non empty complete lattice, $Y(x)$ is an
increasing correspondence (in the induced set ordering) from $\X$ to
the set of the non empty complete 
sublattices of $\X$.  
Then the set of fixed point of $Y$ is a nonempty complete sublattice. 
\end{theorem}

\begin{proof}[Proof of Theorem \ref{th:existence}]
 With Lemma
\ref{lemma:increasingBR}, $Y = (BR^1(\sigma^{-1})\ldots
BR^n(\sigma^{-n}))$ is increasing in the induced set ordering on $\Sigma$. Since $\Sigma$ is a non
empty complete lattice (this comes from the definition of $\B$), Theorem
\ref{theorem:fixedpoint} ensures that the set of fixed points is a
nonempty complete sublattice. The elements of this  set satisfy the
definition of a Nash equilibrium. 
\end{proof}

\textbf{On the example: multinodal case.}
Observe that the reasoning could be extended to a multinodal, non
symmetric setting. One needs to use the strict increasing difference
property with respect to the neighboring nodes.

\section{Uniqueness Sufficient Condition}
\label{sec:uniqueness}

\begin{proof}[Proof of Theorem \ref{th:uniqueness}]
Assume that the equilibria set is not a singleton.
Then since it is a complete lattice, there exists a biggest and a
smallest equilibria in the set. 
We denote by $\bar{\sigma}\in \Sigma$ and
$\underline{\sigma}\in\Sigma$ the strategy profiles of those
equilibria.
The ratio 
\begin{equation*}
\frac{\bar{\sigma}^i[c]}{  \underline{\sigma}^i[c]}
\end{equation*}
is bounded by $b^*/b_*$ and therefore admits a supremum
$\alpha>1$.
By compactness of $\T$  and continuity of the extremal equilibrium
strategies (hypothesis), 
there exist $i^*\in \I$ and $c^*\in \T^{i^*}$ such that 
\begin{equation*}
\frac{\bar{\sigma}^{i^*}[c^*] }{  \underline{\sigma}^{i^*}[c^*]}=\alpha.
\end{equation*}
For any strategy profile $\sigma\in \Sigma$ we denote by
$\alpha\sigma$ the strategy profile defined
by 
\begin{equation*}
(\alpha \sigma)^i (c) = \alpha (\sigma^i (c)) 
\end{equation*}
for any $i\in \I$ and $c\in \T^i$. 

By \eqref{eq:alpha} 
$\alpha \underline{\sigma}$ belongs to $\Sigma$.
Now observe that  for $i\in \I$, $c\in \T$, 
\begin{eqnarray*}
BR^i(\alpha \underline{\sigma}^{-i})[c] = \argmax_{b\in \B^i} \E_{c^{-i}}  (b - c) K^i(
 b,\alpha \underline{\sigma}^{-i}[c^{-i}]    ) \\
= \argmax_{b\in \B^i} \E_{c^{-i}} (b - c) K^i(
  b/\alpha , \underline{\sigma}^{-i}[c^{-i}]    ) \\
= \argmax_{b\in \B^i} \E_{c^{-i}} (b/\alpha - c/\alpha)K^i(
  b/\alpha , \underline{\sigma}^{-i}[c^{-i}]    ) \\
= \alpha \argmax_{u,  \alpha u \in \B^i} \E_{c^{-i}} (u- c/\alpha)  K^i(
  u , \underline{\sigma}^{-i}[c^{-i}]    ) 
\end{eqnarray*}
where we applied the definition of BR, the scaling relation \eqref{as:scaling}
and  a change of variable. 
Now combining the last computation  with \eqref{eq:strictincrease}, 
we have 
\begin{equation*}
BR^i(\alpha \underline{\sigma}^{-i})[c]< \alpha \underline{\sigma}^i[c]
\end{equation*}

The inequality should be understood in the sense that any element of
the LHS set is smaller than the RHS.
Combining the definition of an equilibrium, Lemma
\ref{lemma:increasingBR} on the monotonicity of the  best
replies, and the last relation, we get
\begin{equation*}
  \bar{\sigma}^i[c] \in BR^i(\bar{\sigma}^{-i})[c] \leq BR^i(
\alpha \underline{\sigma}^{-i})[c]<\alpha \underline{\sigma^i}[c]
\end{equation*}
 which is not coherent with the definition of $\alpha$. We conclude
 that the
 equilibrium is unique.
\end{proof}

\section{A dynamic that converges to the smallest Nash Equilibrium}
\label{sec:scheme}

\subsection{Proof of theorem \ref{th:scheme}}

\begin{proof}
First we need to show that  \eqref{system}
has a solution. 
Let $T>0$ and 
$\D_T$  the set of measurable functions $\sigma$ from $[0,T]$ to
$\Sigma$.
On $\D_T$ we  consider the operator $\phi$
such
that, for any 
$\sigma \in \D_T$,
 $\phi_\sigma\in \D_T$ 
and for any $ (i,c,t)\in \I\times\T \times [0,T]$
\begin{equation*}
[\phi_\sigma(t)]^i(c) = c +
\int_0^t \partial_{b}\Pi^i_{[\sigma(s)]^{-i}}([\sigma(s)]^i(c),c)\dd s.
\end{equation*}

Observe that for $T$ small enough, $\D_T$ is stable by $\phi$, which 
is a contracting operator, and, moreover, $\D_T$ associated with the
$||_\infty$ is a closed subset of a
Banach space. Therefore we can apply Picard fixed point theorem and
denote by $\tilde{\sigma}_t$ the associated fixed point. Iterating the
reasoning we can prolong the flow $\tilde{\sigma}_t$ as long as it
stays strictly  below $b^*$.
Using the single crossing property of $\Pi$ (which is conserved
by integration from $\pi$), the fact that $\Pi$ is concave and that the
dynamic cancels at any equilibrium point, we see that
$\tilde{\sigma}_t$ increases and  stays below $\sigma_*$. 
Therefore the flow is defined for all $t\geq 0$, is bounded and
increasing and therefore converges as $t$ goes to $+\infty$ to a
stationary point. 
Using the concavity of $\Pi$, we deduce that the stationary point
is an equilibrium, therefore  the flow converges to $\sigma_*$.  
\end{proof}

\subsection{Remarks}
Observe that we could build the same kind of scheme to compute
$\sigma^*$. 
We introduce this scheme because  it showed better
converging properties on our example than the Best Reply iterations. 

We display in figure \ref{fig:experiment1} and  \ref{fig:experiment2} a numerical experiment with the Best
Reply iterations  and the continuous time dynamic. 
Even when the hypotheses of Theorem \ref{th:scheme} are not satisfied,
the scheme displayed good  converging properties. 
Our intuition is that it maintains the monotonic structure of the problem even when Tarksy's fixed-point theorem hypotheses   are not satisfied.

    \begin{figure}[b]
  \includegraphics[width=\textwidth]{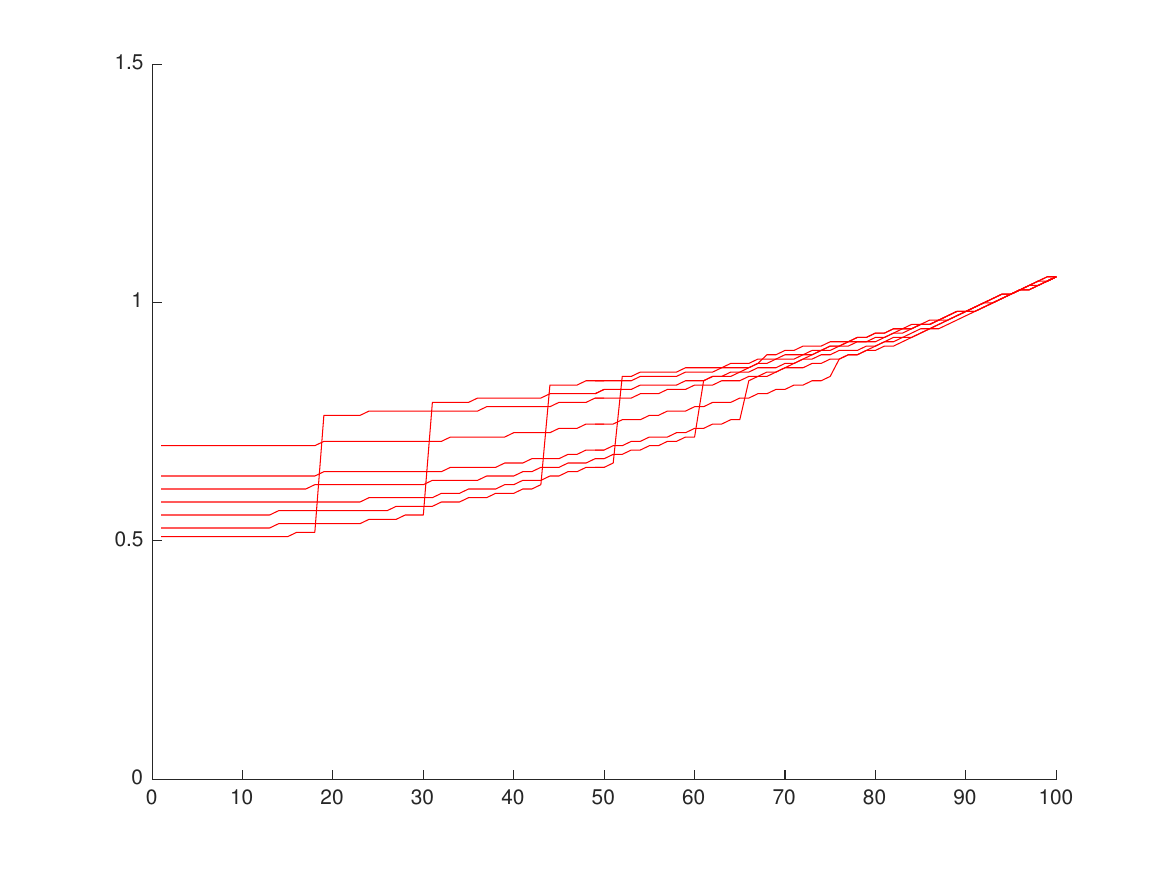}
\caption{The iterated best replies algorithm does not converge}\label{fig:experiment1}
    \end{figure}

    \begin{figure}[b]
        \includegraphics[width=\textwidth]{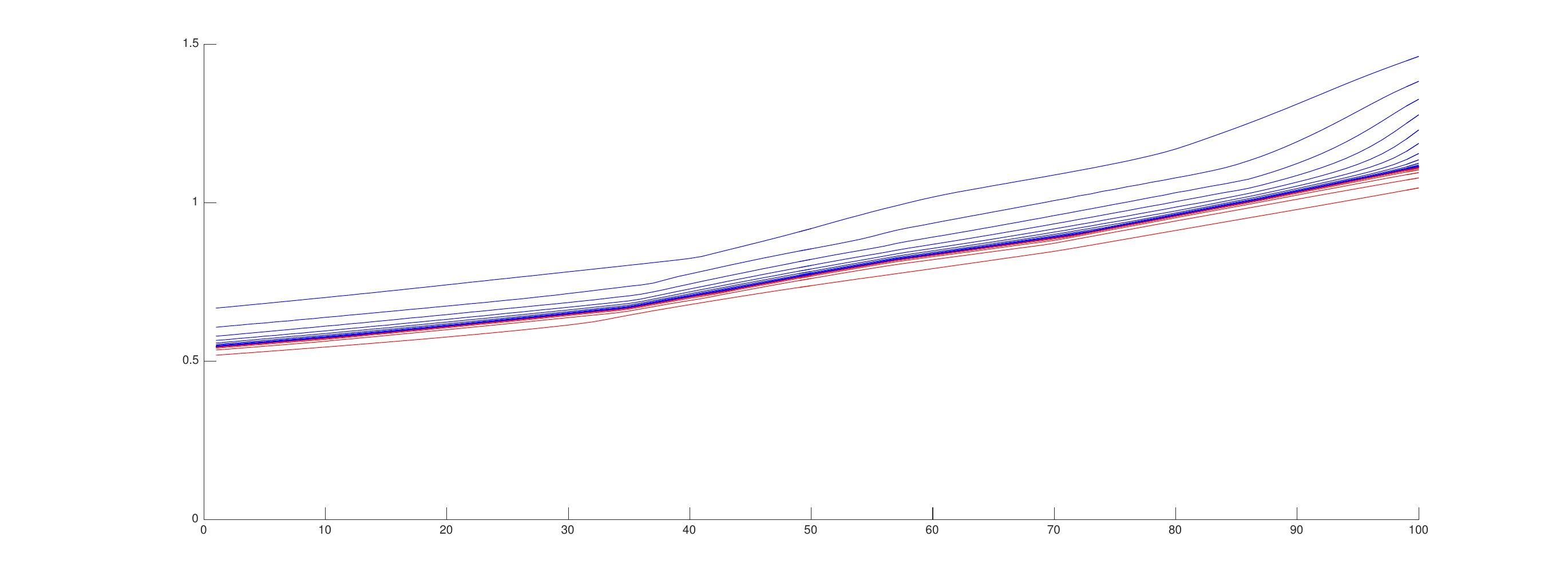}
        \caption{The differential approach converges when we start
          above or below the equilibrium strategy profil} \label{fig:experiment2}
    \end{figure}

\section{Conclusion and possible extension}
We have identified a class  of Bayesian games for which we showed that
there exists a unique pure Nash equilibrium to which a simple dynamic
converges. 
Numerical experiments seem to indicate that those results could be
reinvestigated with weaker assumptions. 
Further works include the extension to
general bidding functions (for instance, by replacing the constant marginal rate with a piece-wise linear one~\cite{heymann:hal-01416411}).

\providecommand{\bysame}{\leavevmode\hbox to3em{\hrulefill}\thinspace}
\providecommand{\MR}{\relax\ifhmode\unskip\space\fi MR }
\providecommand{\MRhref}[2]{%
  \href{http://www.ams.org/mathscinet-getitem?mr=#1}{#2}
}
\providecommand{\href}[2]{#2}


\begin{thebibliography}{10}

\bibitem{arrow60}
Kenneth~J. Arrow and Leonid Hurwicz, \emph{Stability of the gradient process in
  n-person games}, Journal of the Society for Industrial and Applied
  Mathematics \textbf{8} (1960), no.~2, 280--294.

\bibitem{athey2001single}
Susan Athey, \emph{Single crossing properties and the existence of pure
  strategy equilibria in games of incomplete information}, Econometrica
  \textbf{69} (2001), no.~4, 861--889.

\bibitem{bayer2021best}
P{\'e}ter Bayer, Gy{\"o}rgy Kozics, and N{\'o}ra~Gabriella Sz{\H{o}}ke,
  \emph{Best-response dynamics in directed network games}, arXiv preprint
  arXiv:2101.03863 (2021).

\bibitem{berger2008learning}
Ulrich Berger, \emph{Learning in games with strategic complementarities
  revisited}, Journal of Economic Theory \textbf{143} (2008), no.~1, 292--301.

\bibitem{berger2009convergence}
Ulrich Berger et~al., \emph{The convergence of fictitious play in games with
  strategic complementarities: A comment}, MPRA Paper (2009), no.~20241.

\bibitem{brown1949some}
George~W Brown, \emph{Some notes on computation of games solutions}, Tech.
  report, DTIC Document, 1949.

\bibitem{CEPARANO2017154}
Maria~Carmela Ceparano and Federico Quartieri, \emph{Nash equilibrium
  uniqueness in nice games with isotone best replies}, Journal of Mathematical
  Economics \textbf{70} (2017), 154--165.

\bibitem{chenault1986uniqueness}
Larry~A Chenault, \emph{On the uniqueness of nash equilibria}, Economics
  Letters \textbf{20} (1986), no.~3, 203--205.

\bibitem{edlin1998strict}
Aaron~S Edlin and Chris Shannon, \emph{Strict monotonicity in comparative
  statics}, Journal of Economic Theory \textbf{81} (1998), no.~1, 201--219.

\bibitem{fibich2003asymmetric}
Gadi Fibich and Arieh Gavious, \emph{Asymmetric first-price auctions—a
  perturbation approach}, Mathematics of operations research \textbf{28}
  (2003), no.~4, 836--852.

\bibitem{Fibich2011}
Gadi Fibich and Nir Gavish, \emph{Numerical simulations of asymmetric
  first-price auctions}, Games and Economic Behavior \textbf{73} (2011),
  479--495.

\bibitem{gabay1980uniqueness}
Daniel Gabay, \emph{On the uniqueness and stability of nash-equilibria in
  noncooperative games}, Applied Stochastic Control in Economatrics and
  Management Science (1980).

\bibitem{gaudet1991uniqueness}
Gerard Gaudet and Stephen~W Salant, \emph{Uniqueness of cournot equilibrium:
  new results from old methods}, The Review of Economic Studies \textbf{58}
  (1991), no.~2, 399--404.

\bibitem{Gayle2008}
Wayne~Roy Gayle and Jean~Francois Richard, \emph{Numerical solutions of
  asymmetric, first-price, independent private values auctions}, Computational
  Economics \textbf{32} (2008), 245--278.

\bibitem{Goodman1965NoteOE}
J.~Goodman, \emph{Note on existence and uniqueness of equilibrium points for
  concave n-person games}, Econometrica \textbf{48} (1965), 251--251.

\bibitem{NicolasFigueroaAlejandroJofrBenjaminHeymann}
Benjamin Heymann and Alejandro Jofr{\'e}, \emph{Cost-minimizing regulations for
  a wholesale electricity market},  (2015).

\bibitem{heymann:hal-01416411}
\bysame, \emph{{Mechanism design and allocation algorithms for network markets
  with piece-wise linear costs and externalities}}, working paper or preprint,
  December 2016.

\bibitem{lebrun1999first}
Bernard Lebrun, \emph{First price auctions in the asymmetric n bidder case},
  International Economic Review \textbf{40} (1999), no.~1, 125--142.

\bibitem{marshall1994numerical}
Robert~C Marshall, Michael~J Meurer, Jean-Francois Richard, and Walter
  Stromquist, \emph{Numerical analysis of asymmetric first price auctions},
  Games and Economic behavior \textbf{7} (1994), no.~2, 193--220.

\bibitem{maskin2003uniqueness}
Eric Maskin, John Riley, et~al., \emph{Uniqueness of equilibrium in sealed
  high-bid auctions}, Games and Economic Behavior \textbf{45} (2003), no.~2,
  395--409.

\bibitem{mcadams2003isotone}
David McAdams, \emph{Isotone equilibrium in games of incomplete information},
  Econometrica \textbf{71} (2003), no.~4, 1191--1214.

\bibitem{milgrom1990rationalizability}
Paul Milgrom and John Roberts, \emph{Rationalizability, learning, and
  equilibrium in games with strategic complementarities}, Econometrica: Journal
  of the Econometric Society (1990), 1255--1277.

\bibitem{milgrom1991adaptive}
\bysame, \emph{Adaptive and sophisticated learning in normal form games}, Games
  and economic Behavior \textbf{3} (1991), no.~1, 82--100.

\bibitem{milgrom1994comparing}
\bysame, \emph{Comparing equilibria}, The American Economic Review (1994),
  441--459.

\bibitem{milgrom1982theory}
Paul~R Milgrom and Robert~J Weber, \emph{A theory of auctions and competitive
  bidding}, Econometrica: Journal of the Econometric Society (1982),
  1089--1122.

\bibitem{reny2011existence}
Philip~J Reny, \emph{On the existence of monotone pure-strategy equilibria in
  bayesian games}, Econometrica \textbf{79} (2011), no.~2, 499--553.

\bibitem{reny2004existence}
Philip~J Reny and Shmuel Zamir, \emph{On the existence of pure strategy
  monotone equilibria in asymmetric first-price auctions}, Econometrica
  \textbf{72} (2004), no.~4, 1105--1125.

\bibitem{10.2307/1969530}
Julia Robinson, \emph{An iterative method of solving a game}, Annals of
  Mathematics \textbf{54} (1951), no.~2, 296--301.

\bibitem{rosen1965existence}
J~Ben Rosen, \emph{Existence and uniqueness of equilibrium points for concave
  n-person games}, Econometrica: Journal of the Econometric Society (1965),
  520--534.

\bibitem{shannon1995weak}
Chris Shannon, \emph{Weak and strong monotone comparative statics}, Economic
  Theory \textbf{5} (1995), no.~2, 209--227.

\bibitem{tarski1955lattice}
Alfred Tarski et~al., \emph{A lattice-theoretical fixpoint theorem and its
  applications.}, Pacific journal of Mathematics \textbf{5} (1955), no.~2,
  285--309.

\bibitem{topkis1978minimizing}
Donald~M Topkis, \emph{Minimizing a submodular function on a lattice},
  Operations research \textbf{26} (1978), no.~2, 305--321.

\bibitem{topkis1998supermodularity}
\bysame, \emph{Supermodularity and complementarity}, Princeton university
  press, 1998.

\bibitem{van2007monotone}
Timothy Van~Zandt and Xavier Vives, \emph{Monotone equilibria in bayesian games
  of strategic complementarities}, Journal of Economic Theory \textbf{134}
  (2007), no.~1, 339--360.

\bibitem{vives1990nash}
Xavier Vives, \emph{Nash equilibrium with strategic complementarities}, Journal
  of Mathematical Economics \textbf{19} (1990), no.~3, 305--321.

\bibitem{vives2005complementarities}
\bysame, \emph{Complementarities and games: New developments}, Journal of
  Economic Literature \textbf{43} (2005), no.~2, 437--479.

\end{thebibliography}
\end{document}